\documentclass[a4paper]{llncs}

\usepackage{listings}
\usepackage{graphicx}
\usepackage{float,subfig}
\usepackage{xspace,framed}
\usepackage{pgf}
\usepackage[noend]{algpseudocode}
\usepackage{algorithm}
\usepackage{amsfonts}
\usepackage{amsmath}
\usepackage{pifont}
\usepackage{multicol}
\usepackage{multirow}
\usepackage{esvect}

\let\llncssubparagraph\subparagraph
\let\subparagraph\paragraph
\usepackage[compact]{titlesec}
\let\subparagraph\llncssubparagraph

\usepackage{tikz}
\usetikzlibrary{arrows, automata, shapes}

\usepackage{etoolbox}

\makeatletter
\patchcmd{\table}{0\p@}{5\p@}{}{}
\patchcmd{\table}{10\p@}{5\p@}{}{}
\makeatother

\algrenewcommand\algorithmicrequire{\textbf{Precondition:}}
\algrenewcommand\algorithmicensure{\textbf{Postcondition:}}

\newcommand{\newC}{C$^-$\xspace}
\newcommand*\Let[2]{\State #1 $\gets$ #2}

\newcommand{\bv}[2]{\mathcal{BV}(#1, #2)}

\newcommand{\kalashnikov}{{\sc Kalashnikov}\xspace}

\title{Using Program Synthesis for Program Analysis}
\author{Cristina David \and Daniel Kroening \and Matt Lewis}
\institute{University of Oxford}

\makeatletter
\pgfdeclareshape{datastore}{
  \inheritsavedanchors[from=rectangle]
  \inheritanchorborder[from=rectangle]
  \inheritanchor[from=rectangle]{center}
  \inheritanchor[from=rectangle]{base}
  \inheritanchor[from=rectangle]{north}
  \inheritanchor[from=rectangle]{north east}
  \inheritanchor[from=rectangle]{east}
  \inheritanchor[from=rectangle]{south east}
  \inheritanchor[from=rectangle]{south}
  \inheritanchor[from=rectangle]{south west}
  \inheritanchor[from=rectangle]{west}
  \inheritanchor[from=rectangle]{north west}
  \backgroundpath{
    \southwest \pgf@xa=\pgf@x \pgf@ya=\pgf@y
    \northeast \pgf@xb=\pgf@x \pgf@yb=\pgf@y
    \pgfpathmoveto{\pgfpoint{\pgf@xa}{\pgf@ya}}
    \pgfpathlineto{\pgfpoint{\pgf@xb}{\pgf@ya}}
    \pgfpathmoveto{\pgfpoint{\pgf@xa}{\pgf@yb}}
\usepackage{tikz}
\usetikzlibrary{arrows, automata, shapes}

    \pgfpathlineto{\pgfpoint{\pgf@xb}{\pgf@yb}}
 }
}
\makeatother

\begin{document}
\maketitle
\pagestyle{headings}  

\begin{abstract}
In this paper, we identify a fragment of second-order logic with restricted quantification that is 
expressive enough to capture numerous static analysis problems (e.g. safety proving, bug finding, termination and
non-termination proving, superoptimisation). 
We call this fragment the {\it synthesis fragment}.
Satisfiability of a formula in the synthesis fragment is decidable over finite domains; specifically the decision problem is NEXPTIME-complete.
If a formula in this fragment is satisfiable, a solution consists of a satisfying assignment from
the second order variables to \emph{functions over finite domains}.
To concretely find these solutions, we synthesise \emph{programs} that compute the functions.
Our program synthesis algorithm is complete for finite state programs, i.e. every \emph{function} over
finite domains is computed by some \emph{program} that we can synthesise. 
We can therefore use our synthesiser as a decision procedure for the synthesis fragment of second-order logic,
which in turn allows us to use it as a powerful backend for many program analysis tasks.
To show the tractability of our approach, we evaluate the program synthesiser on several static analysis problems.

\end{abstract}


\section{Introduction}
Fundamentally, every static program analysis is searching for a {\em program proof}.
For safety analysers this proof takes the form of 
a program invariant \cite{DBLP:conf/popl/CousotC77}, for bug finders it's a counter-model \cite{DBLP:conf/dac/ClarkeKY03}, for termination analysis it 
can be a ranking function \cite{Floyd67}, whereas for non-termination it's a recurrence set \cite{DBLP:conf/popl/GuptaHMRX08}.
Finding each of these proofs was subject to extensive research resulting in a multitude of techniques.

The process of searching for a proof can be roughly seen as a refinement loop with 
two phases. One phase is heuristic in nature, e.g. 
adjusting the unwinding depth (for bounded model checking \cite{DBLP:conf/dac/ClarkeKY03}),
refining the set of predicates (for predicate abstraction and interpolation \cite{DBLP:conf/cav/ClarkeGJLV00,DBLP:conf/cav/McMillan06}),
selecting a template (for template-based analyses \cite{dblp:conf/tacas/leikeh14}),
applying a widening operator (for abstract interpretation based techniques \cite{DBLP:conf/popl/CousotC77}),
whereas the other phase usually involves a call to a decision procedure. 
From the perspective of the proof,
the heuristic constrains the universe of potential proofs to just one candidate 
(by fixing the unwinding bound, the template, the set of predicates, etc), which 
is then validated by the other phase.
The unknowns in the first phase are proofs (second-order entities), 
whereas the unknowns in the second phase are program variables.
Essentially, the first phase reduces a second-order problem to a first-order/propositional one.

Such a design makes it difficult to separate the problem's formulation (the second-order problem)
from parts of the solving process, resulting in 
analyses that are cluttered and fragile. 
Any change to the search process causes changes to the whole analysis.
Ideally, we would like a modular design, where the search for a solution is encapsulated and, thus, completely separated from the formulation of the problem.

The existing design is dictated by the state of the art in solver technology.
While the SAT/SMT technologies nowadays allow solving industrial sized instances, 
there is hardly any progress made for second-order solvers.


\paragraph{Our Contributions.}
In this paper, we take a step towards addressing this situation by 
identifying a fragment of second-order logic with restricted quantification that is 
expressive enough to capture numerous static analysis problems (e.g. safety proving, bug finding, termination and
non-termination proving, superoptimisation) and solving it via {\it program synthesis}.
We call this fragment the {\it synthesis fragment}.
The synthesis fragment is decidable over finite domains, and its decision problem is NEXPTIME-complete.

If a formula in the synthesis fragment is satisfiable, a solution consists of a satisfying assignment from
the second order variables to \emph{functions over finite domains}.
%
Additionally, 
every function over
finite domains is computed by some \emph{program} that we can synthesise.  This correspondence between
logical satisfiability and program synthesis allows us to design the decision procedure
for the synthesis fragment as a program synthesiser, which in turn allows us to use it as a powerful backend for
many program analysis tasks.

The program synthesiser uses a combination of bounded model checking, explicit state
model checking and genetic programming. 
Additionally, one of its important strategies is {\it generalisation} -- it finds simple
solutions that solve a restricted case of the problem, then it tries to generalise to
a full solution. 
Consequently, our approach achieves surprisingly good
performance for a problem with such high complexity,
as whown by our experimental evaluation 
on several static analysis problems. 



\paragraph{Our program synthesiser vs Syntax Guided Synthesis (SyGuS) \cite{sygus}.}
Recently, a successful approach to program synthesis is denoted by SyGuS. The SyGuS synthesisers
supplement the logical specification with a syntactic template that constrains the space of allowed implementations.
Thus, each semantic specification is accompanied by a syntactic specification in the form of a grammar. 

In contrast to SyGuS, our program synthesiser is optimised for program analysis in the following ways:
\begin{enumerate}
\item Our specification language is a fragment of C, which results in concise specifications.  Using our tool to build a program analyser only requires providing 
a generic specification of the problem to solve.  The programs to be analysed do not need to be modified, symbolically executed or compiled to an intermediate language.
As shown by our experimental results in Sec~\ref{sec:kalashnikov-experiments}, our specifications are an order of magnitude smaller than the equivalent SyGuS specifications.
\item We do not require specifications to include a grammar restricting the space of solutions.  The language in which we synthesise our programs is universal, i.e.
every finite function is computed by at least one program in our language.  Furthermore this language is concise, which means that programs in this language tend to be
short.  This property of not requiring a grammar makes it easy to target our specifications from an automated program analysis tool, since the frontend analyser
need not guess the syntax of a solution.
\item Our solution language has first-class support for programs computing multiple outputs, which allows us to directly encode lexicographic ranking functions of
unbounded dimension.
We were unable to find a way to encode such functions with SyGuS.
\item Our solver has first-class support for synthesising programs containing constants.  By contrast, SyGuS builds constants using unary encoding which requires
exponentially more space.  We discuss this in Section~\ref{sec:kalashnikov-param-space}.
\item Our solver's runtime is heavily influenced by the {\em length of the shortest
proof}, i.e. the Kolmogorov complexity of the problem as we will discuss in Sec~\ref{sec:kalashnikov-proofs}.
If a short proof exists, then the solver will find it quickly. 
This is particularly useful for program analysis problems,
where, if a proof exists, then most of the time many proofs exist and some are short 
(\cite{DBLP:conf/aplas/KongJDWY10} relies on a similar remark about loop invariants).
\end{enumerate}




\paragraph{Related logics.}
Other second-order solvers are introduced in \cite{DBLP:conf/pldi/GrebenshchikovLPR12,DBLP:conf/cav/BeyenePR13}. As opposed to ours,
these are specialised for Horn clauses and the logic they handle is undecidable.
Wintersteiger et al. present in \cite{DBLP:conf/fmcad/WintersteigerHM10} a decision procedure for  
a logic related to the synthesis fragment, the Quantified bit-vector logic, which is  
a many sorted first-order logic formula where the sort of every variable is a bit-vector sort. 
It is possible to reduce formulae in the synthesis fragment over finite domains to Effectively Propositional Logic \cite{DBLP:journals/jar/PiskacMB10}, 
but the reduction would require additional axiomatization and would 
increase the search space, thus defeating the efficiency we are
aiming to achieve.

\section{The Synthesis Fragment}

In this section, we identify a fragment of second-order logic
with a constrained use of quantification that is expressive enough to encode
numerous static analysis problems. 
We will suggestively refer to the fragment as the \emph{synthesis fragment}:

\begin{definition}[Synthesis Fragment ($SF$)] \label{def:synthesis-frag}
A formula is in the \emph{synthesis fragment} iff it is of the form
 \[
  \exists P_1 \ldots P_m . Q_1 x_1 \ldots Q_n x_n . \sigma(P_1,\ldots,P_m, x_1,\ldots, x_n)
 \]
where the $P_i$ range over functions, 
the $Q_i$ are either $\exists$ or $\forall$,
the $x_i$ range over ground terms and $\sigma$ is a quantifier-free formula.

%
\end{definition}
%

If a pair $(\vv{P}, \vv{x})$ is a satisfying model for the synthesis formula, then we write $(\vv{P}, \vv{x}) \models \sigma$.
For the remainder of the presentation, we drop the vector notation and write $x$ for $\vv{x}$, with the understanding
that all quantified variables range over vectors.






\section{Program Analysis Specifications in the Synthesis Fragment}

Program analysis  problems can be reduced to the problem of finding
solutions to a second-order constraint \cite{DBLP:conf/pldi/GulwaniSV08,DBLP:conf/pldi/GrebenshchikovLPR12,synth-termination}.
In this section we will show that, the synthesis fragment is expressive enough to capture 
many interesting such problems. 
When we describe analyses involving loops, we will characterise the loop
as having initial state $I$, guard $G$, transition relation $B$.

\paragraph{Safety Invariants.}
Given a safety assertion $A$, a safety invariant is a set of states $S$ which is inductive with respect to the program's transition
relation, and which excludes an error state.  A predicate $S$ is a safety invariant iff it satisfies the following criteria:%
\begin{align}
  \exists S . \forall x, x' .  & I(x) \rightarrow S(x) ~ \wedge \label{safety_base}\\
  & S(x) \wedge G(x) \wedge B(x, x') \rightarrow S(x') ~ \wedge \label{safety_inductive}\\
  & S(x) \wedge \neg G(x) \rightarrow A(x) \label{safety_safe}
\end{align}
(\ref{safety_base}) says that each state reachable on entry to the loop is in the set $S$, and in
combination with (\ref{safety_inductive}) shows that every state that can be reached by the loop
is in $S$.  The final criterion~(\ref{safety_safe}) says that if the loop exits while in an $S$-state,
the assertion $A$ is not violated.  

\paragraph{Termination and non-termination.}
As shown in~\cite{synth-termination}, termination of a loop can be encoded as the following
formula, where $W$ is an inductive invariant of the loop that is established by the initial states $I$ if the
loop guard $G$ is met, and $R$ is a ranking function as restricted by $W$: 
%
 \begin{align*}
  \exists R, W . \forall x, x' . & I(x) \wedge G(x) \rightarrow W(x) ~ \wedge \\
                                 & G(x) \wedge W(x) \wedge B(x, x') \rightarrow W(x') \wedge R(x) > 0 
  \wedge R(x) > R(x')
 \end{align*}
%
Similarly, non-termination can be expressed in the synthesis fragment as follows:
%
\label{def:snt}
  $$\exists N, C, x_0 . \forall x .  N(x_0) ~\wedge~    N(x) \rightarrow G(x) ~ \wedge 
		 N(x) \rightarrow B(x, C(x)) \wedge N(C(x))$$

%
$N$ denotes a recurrence set, i.e. a nonempty set of
states such that for each $s \in N$ there exists a transition to some $s' \in N$, and
$C$ is a Skolem function that chooses the successor $x'$.
More details on the formulations for termination and non-termination can be found in \cite{synth-termination}.




\section{The Synthesis Fragment over Finite Domains} \label{sec:finite.domain}

When interpreting the ground terms over a finite domain $\mathcal{D}$, the synthesis fragment is decidable
and its 
decision problem is NEXPTIME-complete.
We use the notation $SF_\mathcal{D}$ to denote the synthesis fragment over a finite domain $\mathcal{D}$. 


\begin{theorem}[$SF_\mathcal{D}$ is NEXPTIME-complete]
 For an instance of Definition~\ref{def:synthesis-frag} with $n$ first-order variables, where the ground terms are interpreted over $\mathcal{D}$,
 checking the truth of the formula is NEXPTIME-complete.
\end{theorem}

\begin{proof}
In Appendix~\ref{app:finite-domain}.
\end{proof}



Next, we are concerned with building a solver for $SF_\mathcal{D}$.
%
%
A satisfying model for a formula in $SF_\mathcal{D}$ 
is an assignment mapping
each of the second-order variables to some function of the appropriate type and arity.
When deciding whether a particular $SF_\mathcal{D}$ instance is satisfiable, we should think
about how solutions are encoded and in particular how a function is to be encoded.
The functions all have a finite domain and co-domain, so their canonical representation
would be a finite set of ordered pairs.  Such a set is exponentially large in the size of
the domain, so we would prefer to work with a more compact representation if possible.

We will generate \emph{finite state programs} that compute the functions and represent these
programs as finite lists of instructions in SSA form.  This representation has the following
properties, proofs for which can be found in Appendix~\ref{app:encodings}.

\begin{theorem}
\label{thm:l-universal}
Every total, finite function is computed by at least one program in this language
\end{theorem}
\begin{theorem}
\label{thm:l-concise}
Furthermore, this representation is optimally
concise -- there is no encoding that gives a shorter representation to every function.
\end{theorem}

\paragraph{Finite State Program Synthesis}

To formally define the finite state synthesis problem, we need to fix some notation.
We will say that a program $P$ is a finite list of instructions in SSA form, where no
instruction can cause a back jump, i.e. our programs are loop free and non-recursive.
Inputs $x$ to the program are drawn from some finite domain $\mathcal{D}$.
The synthesis problem is given to us in the form of a specification $\sigma$ which is
a function taking a program $P$ and input $x$ as parameters and returning a boolean
telling us whether $P$ did ``the right thing'' on input $x$.
Basically, the finite state synthesis
problem checks the truth of Definition~\ref{def:finite-synth-formula}.

\begin{definition}[Finite Synthesis Formula]
 \label{def:finite-synth-formula}
\[
 \exists P . \forall x \in \mathcal{D} . \sigma(P, x)
\]
\end{definition}

To express the specification $\sigma$, we introduce a function
$\mathtt{exec}(P, x)$ that returns the result of running program $P$ with input $x$.
Since $P$ cannot contain loops or recursion, $\mathtt{exec}$ is a total function.

\begin{example}
The following finite state synthesis problem is satisfiable:
\begin{center}
$ \exists P . \forall x \in \mathbb{N}_8 . \mathtt{exec}(P, x) \geq x$
\end{center}

One such program $P$ satisfying the specification is \verb|return 8|, which just returns 8 for any input.
\end{example}

We now present our main theorem, which says that satisfiability of $SF_\mathcal{D}$ can be reduced to finite state
program synthesis.  The proof of this theorem can be found in
Appendix~\ref{app:proofs}.

\begin{theorem}[$SF_\mathcal{D}$ is Polynomial Time Reducible to Finite Synthesis]
 Every instance of 
Definition~\ref{def:synthesis-frag}, where the ground terms are interpreted over $\mathcal{D}$
 is polynomial time reducible to an instance
 of Definition~\ref{def:finite-synth-formula}.
\end{theorem}

\begin{corollary}
\label{cor:finite-synth-nexp}
 Finite-state program synthesis is NEXPTIME-complete.
\end{corollary}

We are now in a position to sketch the design of a decision procedure for $SF_\mathcal{D}$: we will convert
the $SF_\mathcal{D}$ satisfiability problem to an equisatisfiable finite synthesis problem which we will then solve
with a finite state program synthesiser.  This design will be elaborated in Section~\ref{sec:kalashnikov-implementation}.

\section{Deciding $SF_\mathcal{D}$ via Finite-State Program Synthesis}
\label{sec:kalashnikov-implementation}


In this section we will present a sound and complete algorithm for
finite-state synthesis that we use to decide the satisfiability of formulae in $SF_\mathcal{D}$.
We begin by describing a general purpose synthesis procedure (Section~\ref{sec:kalashnikov-abstract-synth}),
then detail how this general purpose procedure is instantiated for
synthesising finite-state programs.  
We then describe the algorithm
we use to search the space of possible programs (Sections~\ref{sec:kalashnikov-candidates}, \ref{sec:kalashnikov-param-space} and~\ref{sec:kalashnikov-search-space}).

\subsection{General Purpose Synthesis Algorithm}
\label{sec:kalashnikov-abstract-synth}

\begin{algorithm}
 \caption{Abstract refinement algorithm
 \label{alg:cegis}}
 \begin{multicols}{2}
 \begin{algorithmic}[1]
\Statex
\Function{synth}{inputs}
  \Let{$(i_1, \ldots, i_N)$}{inputs}
  \Let{query}{$\exists P . \sigma(i_1, P) \land \ldots \land \sigma(i_N, P)$}
  \Let{result}{decide(query)}
  \If{result.satisfiable}
    \State \Return{result.model}
  \Else
    \State \Return{UNSAT}
  \EndIf
\EndFunction
\Statex
\Function{verif}{P}
  \Let{query}{$\exists x . \lnot \sigma(x, P)$}
  \Let{result}{decide(query)}
  \If{result.satisfiable}
    \State \Return{result.model}
  \Else
    \State \Return{valid}
  \EndIf
\EndFunction
\columnbreak
\Statex
\Function{refinement loop}{}
  \Let{inputs}{$\emptyset$}
  \Loop
    \Let{candidate}{\Call{synth}{inputs}}
    \If{candidate = UNSAT}
      \State \Return{UNSAT}
    \EndIf
    \Let{res}{\Call{verif}{candidate}}
    \If{res = valid}
      \State \Return{candidate}
    \Else
      \Let{inputs}{inputs $\cup$ res}
    \EndIf
  \EndLoop
\EndFunction
 \end{algorithmic}
 \end{multicols}
\end{algorithm}

We use Counterexample Guided Inductive Synthesis
(CEGIS)~\cite{sketch,toast} to find a
program satisfying our specification.  
Algorithm~\ref{alg:cegis} is divided into two procedures: {\sc synth} 
and {\sc verif}, which interact via a finite set of test vectors {\sc inputs}.
The {\sc synth} procedure tries to find an existential witness $P$
that satisfies the partial specification:
$\exists P . \forall x \in \text{\sc inputs} . \sigma(x, P)$

If {\sc synth} succeeds in finding a witness $P$, this witness is a
candidate solution to the full synthesis formula.  We pass this candidate
solution to {\sc verif} which determines whether it does satisfy
the specification on all inputs by checking satisfiability of the
verification formula:
$\exists x . \lnot \sigma(x, P)$

If this formula is unsatisfiable, the candidate solution is in fact a
solution to the synthesis formula and so the algorithm terminates. 
Otherwise, the witness $x$ is an input on which the candidate solution fails
to meet the specification.  This witness $x$ is added to the {\sc inputs}
set and the loop iterates again.  It is worth noting that each iteration of the
loop adds a new input to the set of inputs being used for synthesis.  If
the full set of inputs is finite, this means that the refinement loop
can only iterate a finite number of times.

\subsection{Finite-State Synthesis}
\label{sec:kalashnikov-concrete-algorithm}

We will now show how the generic construction of Section~\ref{sec:kalashnikov-abstract-synth}
can be instantiated to produce a finite-state program synthesiser.
A natural choice for such a synthesiser would be to
work in the logic of quantifier-free propositional formulae and to use a
propositional SAT or SMT-$\mathcal{BV}$ solver as the decision procedure. 
However we propose a slightly different tack, which is to use a decidable
fragment of C as a ``high level'' logic.  We call this fragment \newC.


The characteristic property of a
\newC  program is that safety can be decided for
it using a single query to a Bounded Model Checker.  A \newC program is
just a C program with the following syntactic restrictions:\\
(i) all loops in the program must have a constant bound;\\
(ii) all recursion in the program must be limited to a constant depth;\\
(iii) all arrays must be statically allocated (i.e. not using \texttt{malloc}),
 and be of constant size.\\
\newC programs may use nondeterministic values, assumptions
and arbitrary-width types.

Since each loop is bounded by a constant, and each recursive function call is
limited to a constant depth, a \newC program necessarily terminates and in
fact does so in $O(1)$ time.  If we call the largest loop bound~$k$, then
a Bounded Model Checker with an unrolling bound of $k$ will be a complete
decision procedure for the safety of the program.  For a \newC program of
size $l$ and with largest loop bound~$k$, a Bounded Model Checker will
create a SAT problem of size $O(lk)$.  Conversely, a SAT problem
of size $s$ can be converted trivially into a loop-free \newC program
of size $O(s)$.  The safety problem for \newC is therefore NP-complete,
which means it can be decided fairly efficiently for many practical
instances.

\subsection{Candidate Generation Strategies}
\label{sec:kalashnikov-candidates}
 A candidate solution $P$ is written in a simple RISC-like language $\mathcal{L}$, whose syntax is given in Fig.~\ref{fig:l-language} in Appendix~\ref{app:cegis}.  
 We supply an interpreter for $\mathcal{L}$ which is written in \newC.  
The specification function $\sigma$ will include calls to this interpreter, by which means it
will examine the behaviour of a candidate $\mathcal{L}$ program.

For the {\sc synth} portion of the CEGIS loop, we construct a \newC program {\sc synth.c}
which takes as parameters a candidate program $P$ and test
inputs. 
The program contains an assertion which fails
iff $P$ meets the specification for each of the inputs. 
Finding a new candidate program is then equivalent to checking
the safety of {\sc synth.c} for which we use the strategies described in the next section.
There are many possible strategies for finding these candidates; we employ the
following strategies in parallel:

{\it (i) Explicit Proof Search.} The simplest strategy for finding candidates
is to just exhaustively enumerate them all, starting with the shortest and
progressively increasing the number of instructions.  

{\it (ii) Symbolic Bounded Model Checking.} Another complete method for generating
candidates is to simply use BMC on the {\sc synth.c} program.  

{\it (iii) Genetic Programming and Incremental Evolution.} \label{sec:gp}
Our final strategy is genetic programming
(GP)~\cite{langdon:fogp,brameier2007linear}.  GP provides an adaptive way of
searching through the space of $\mathcal{L}$-programs for an individual
that is ``fit'' in some sense.  We measure the
fitness of an individual by counting the number of tests in {\sc inputs}
for which it satisfies the specification.

To bootstrap GP in the first iteration of the CEGIS loop, we generate a population
of random $\mathcal{L}$-programs. We then iteratively evolve this population by
applying the genetic operators {\sc crossover} and {\sc mutate}.
{\sc Crossover} combines selected existing programs into new programs,
whereas {\sc mutate} randomly changes parts of a single program.
Fitter programs are more likely to be selected.



Rather than generating a random population at the beginning of each subsequent
iteration of the CEGIS loop, we start with the population we had at the end of the
previous iteration.  The intuition here is that this population contained
many individuals that performed well on the $k$ inputs we had before, so
they will probably continue to perform well on the $k+1$ inputs we have now.
In the parlance of evolutionary programming, this is known as
incremental evolution~\cite{Gomez97incrementalevolution}.

\section{Searching the Space of Possible Solutions}
\label{sec:kalashnikov-param-space}

An important aspect of our synthesis algorithm is the manner in which we search the space of candidate programs. 
The key component is parametrising
the language~$\mathcal{L}$, which induces a lattice of progressively
more expressive languages.  We start by attempting to synthesise
a program at the lowest point on this lattice and increase the
parameters of~$\mathcal{L}$ until we reach a point at which
the synthesis succeeds.

As well as giving us an automatic search procedure, this parametrisation
greatly increases the efficiency of our system since languages
low down the lattice are very easy to decide safety for.  If a program
can be synthesised in a low-complexity language, the whole procedure
finishes much faster than if synthesis had been attempted in a
high-complexity language.


\subsection{Parameters of language $\mathcal{L}$}


%
{\it Program Length: $l$.}
The first parameter we introduce is program length, denoted by $l$.
At each iteration we synthesise programs of length exactly $l$.
We start with $l = 1$ and increment $l$ whenever we determine
that no program of length $l$ can satisfy the specification.  When we do
successfully synthesise a program, we are \emph{guaranteed that it
is of minimal length} since we have previously established that no
shorter program is correct.\\
{\it Word Width: $w$.}
An $\mathcal{L}$-program runs on a virtual machine (the $\mathcal{L}$-machine) that is parametrised 
by the \emph{word width},  
that is, the number of bits
in each internal register and immediate constant. \\ 
{\it Number of Constants: $c$.}
Instructions in $\mathcal{L}$ take up to three operands.
Since any instruction whose operands are all constants can always be
eliminated (since its result is a constant), we know that a loop-free program
of minimal length will not contain any instructions with two constant
operands.  Therefore the number of constants that can appear in
a minimal program of length $l$ is at most $l$.  By minimising the number
of constants appearing in a program, we are able to use a particularly
efficient program encoding that speeds up the synthesis procedure
substantially.  

\subsection{Searching the Program Space}
\label{sec:kalashnikov-search-space}

The key to our automation approach is to come up with a sensible way in which to
adjust the $\mathcal{L}$-parameters in order to cover all possible programs.
Two important components in this search are the adjustment of parameters and the generalisation of candidate solutions.
We discuss them both next.

\paragraph{Adjusting the search parameters.}
After each round of {\sc synth}, we may need to adjust the parameters.  The
logic for these adjustments is shown as a tree in Fig.~\ref{fig:paramsflow}.

Whenever {\sc synth} fails, we consider which parameter might have caused the
failure.  There are two possibilities: either the program length $l$ was too small,
or the number of allowed constants $c$ was.  If $c < l$, we just increment $c$ and
try another round of synthesis, but allowing ourselves an extra program constant.
If $c = l$, there is no point in increasing $c$ any further.  This is because
no minimal $\mathcal{L}$-program has $c > l$, for if it did there would
have to be at least one instruction with two constant operands.  This
instruction could be removed (at the expense of adding its result as
a constant), contradicting the assumed minimality of the program.  So
if $c = l$, we set $c$ to 0 and increment $l$, before attempting
synthesis again.

If {\sc synth} succeeds but {\sc verif} fails, we have a candidate
program that is correct for some inputs but incorrect on at least
one input.  However, it may be the case that the candidate program
is correct for \emph{all} inputs when run on an $\mathcal{L}$-machine
with a small word size. Thus, we try to generalise the solution to a bigger word size, as explained in the next paragraph. 
If the generalisation
is able to find a correct program, we are done.  Otherwise,
we need to increase the word width of the $\mathcal{L}$-machine
we are currently synthesising for.


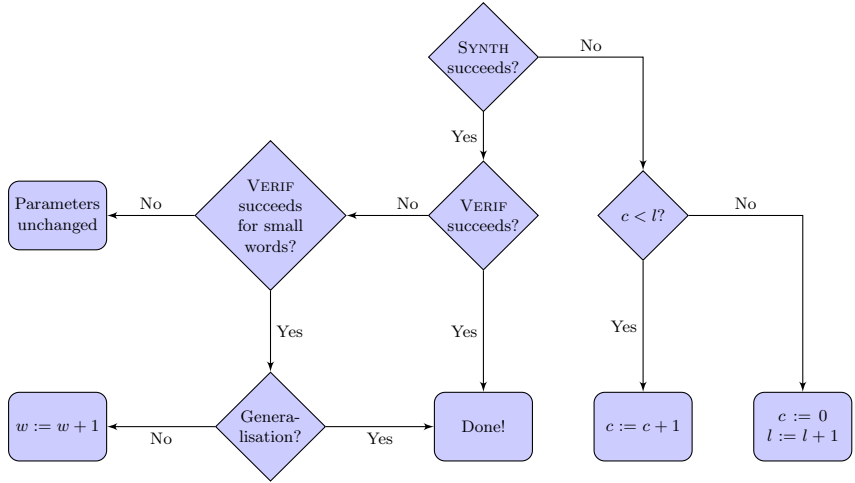
\begin{figure}[t]
\centering
\begin{tikzpicture}[scale=0.7, transform shape, node distance=3cm, auto]
\tikzstyle{decision} = [diamond, draw, fill=blue!20, 
    text width=4.5em, text badly centered, node distance=3cm, inner sep=0pt]
\tikzstyle{block} = [rectangle, draw, fill=blue!20, 
    text width=5em, text centered, rounded corners, minimum height=4em]
\tikzstyle{line} = [draw, -latex']
\tikzstyle{cloud} = [draw, ellipse,fill=red!20, node distance=3cm,
    minimum height=2em]

 \node [decision] (synthsucc) {{\sc Synth} succeeds?};

 \node [decision, below of=synthsucc, node distance=3cm] (verif) {{\sc Verif} succeeds?};
 \node [decision, right of=verif] (ck) {$c < l$?};

 \node [block, below of=verif, node distance=4cm]   (done) {Done!};
 \node [decision, left of=verif, node distance=4cm] (verifw) {{\sc Verif} succeeds for small words?};

 \node [block, below of=ck, node distance=4cm] (incc) {$c := c+1$};
 \node [block, right of=incc, node distance=3cm] (incl) {$c := 0$\\ $l := l+1$};

 \node [decision, below of=verifw, node distance=4cm] (gen) {Genera-\\lisation?};
 \node [block, left of=verifw, node distance=4cm] (iterate) {Parameters unchanged};

 \node [block, left of=gen, node distance=4cm] (incw) {$w := w+1$};

 \path [line] (synthsucc) -- node [left] {Yes} (verif);
 \path [line] (synthsucc) -| node [above, near start] {No} (ck);

 \path [line] (verif) -- node [left] {Yes} (done);
 \path [line] (verif) -- node [above, near start] {No} (verifw);

 \path [line] (ck) -- node [left] {Yes} (incc);
 \path [line] (ck) -| node  [above, near start]  {No} (incl);

 \path [line] (verifw) -- node [right] {Yes} (gen);
 \path [line] (verifw) -- node [above] {No} (iterate);
 
 \path [line] (gen) -- node [below] {Yes} (done);
 \path [line] (gen) -- node [below] {No} (incw);

\end{tikzpicture}

 \caption{Decision tree for increasing parameters of $\mathcal{L}$.}
 \label{fig:paramsflow}

\end{figure}

\paragraph{Generalisation of candidate solutions.}
It is often the case that a program which satisfies the specification
on an $\mathcal{L}$-machine with $w = k$ will continue to satisfy the
specification when run on a machine with $w > k$.  For example, the program
in Fig.~\ref{fig:bitvector-program} isolates the least-significant bit of a word.
This is true irrespective of the word size of the machine it is run on -- it will
isolate the least-significant bit of an 8-bit word just as well as it will a
32-bit word.  An often successful strategy is to synthesise a program for an
$\mathcal{L}$-machine with a small word size and then to check whether the
same program is correct when run on an $\mathcal{L}$-machine with a
full-sized word.

The only wrinkle here is that we will sometimes synthesise a program containing
constants.  If we have synthesised a program with $w=k$,
the constants in the program will be $k$-bits wide.  To extend the program
to an $n$-bit machine (with $n > k$), we need some way of deriving $n$-bit-wide
numbers from $k$-bit ones.  We have several strategies for this and
just try each in turn.  Our strategies are shown in Fig.~\ref{fig:generalize}.
$\mathcal{BV}(v, n)$ denotes an $n$-bit wide bitvector holding the value $v$
and $b \cdotp c$ means the concatenation of bitvectors $b$ and $c$.  For
example, the first rule says that if we have the 8-bit number with value 8,
and we want to extend it to some 32-bit number, we'd try the 32-bit number
with value 32.  These six rules are all heuristics that we have found to be
fairly effective in practice.

\begin{figure}
\centering
\begin{minipage}{0.45\linewidth}
 \begin{lstlisting}[language=C]
int isolate_lsb(int x) {
  return x & -x;
}
 \end{lstlisting}
\end{minipage}
\begin{minipage}{0.45\linewidth}

Example:

\hrule

\begin{tabular}{llcccccccc}
 x       & = & 1 & 0 & 1 & 1 & 1 & 0 & 1 & 0 \\
 -x      & = & 0 & 1 & 0 & 0 & 0 & 1 & 1 & 0 \\
 x \& -x & = & 0 & 0 & 0 & 0 & 0 & 0 & 1 & 0
\end{tabular}
\end{minipage}

 \caption{A tricky bitvector program}
  \label{fig:bitvector-program}
\end{figure}

\begin{figure}
\centering
\begin{minipage}[t]{.45\textwidth}
\begin{eqnarray*}
 \bv{m}{m} & \rightarrow & \bv{n}{n} \\
 \bv{m-1}{m} & \rightarrow & \bv{n-1}{n} \\
 \bv{m+1}{m} & \rightarrow & \bv{n+1}{n}
\end{eqnarray*}
\end{minipage}
\begin{minipage}[t]{.45\textwidth}
\begin{eqnarray*}
 \bv{x}{m} & \rightarrow & \bv{x}{n} \\
 \bv{x}{m} & \rightarrow & \bv{x}{m} \cdotp \bv{0}{n - m} \\
 \bv{x}{m} & \rightarrow & \underbrace{\bv{x}{m} \cdotp \ldots \cdotp \bv{x}{m}}_{\frac{n}{m} \mathrm{ times}}
\end{eqnarray*}
\end{minipage}
\caption{Rules for extending an $m$-bit wide number to an $n$-bit wide one.
 \label{fig:generalize}}
\end{figure}

\subsection{Stopping Condition for Unsatisfiable Specifications}
\label{sec:stopping-condition}
If a specification is unsatisfiable, we would still like our algorithm to terminate
with an ``unsatisfiable'' verdict.  To do this, we can observe that any total function
taking $n$ bits of input is computed by some program of at most $2^n$ instructions
(a consequence of Theorems~\ref{thm:l-universal} and~\ref{thm:l-concise}).
Therefore every satisfiable specification has a solution with at most $2^n$ instructions.
This means that if we ever need to increase the length of the candidate program we
search for beyond $2^n$, we can terminate, safe in the knowledge that the
specification is unsatisfiable.




\section{Soundness and Completeness}
\label{sec:kalashnikov-proofs}
We will now state soundness and completeness results for the $SF_\mathcal{D}$ solver.
Proofs for each of these theorems can be found in Appendix~\ref{app:proofs}.

\begin{theorem}\label{thm:synth-sound}
Alg~\ref{alg:cegis} is sound -- if it terminates with witness $P$, then
$P \models \sigma$.
\end{theorem}


\begin{theorem}
 \label{thm:finite-synth-complete}
 Alg~\ref{alg:cegis} with the stopping condition described in Section~\ref{sec:stopping-condition}
 is complete when instantiated with \newC as a background theory -- it will terminate for all specifications $\sigma$.
\end{theorem}

 Since safety of \newC programs is decidable, Algorithm~\ref{alg:cegis} is sound and complete
 when instantiated with \newC as a background theory and using the stopping condition of
 Section~\ref{sec:stopping-condition}.  This construction therefore gives as a decision procedure
 for $SF_\mathcal{D}$.

\paragraph{Runtime as a Function of Solution Size.}
We note that the runtime of our solver is heavily influenced by the length of the shortest
program satisfying the specification, since we begin searching for short programs.
%
%
We will now show that the number of iterations of the CEGIS loop
is a function of the Kolmogorov complexity of the synthesised program.
Let us first recall the definition of the Kolmogorov complexity
of a function $f$:

\begin{definition}[Kolmogorov complexity]
 The Kolmogorov complexity $K(f)$ is the length of the shortest program that 
 computes~$f$.
\end{definition}

We can extend this definition slightly to talk about the Kolmogorov complexity of a
synthesis problem in terms of its specification:

\begin{definition}[Kolmogorov complexity of a synthesis problem]
 The Kolmogorov complexity of a program specification $K(\sigma)$ is the length of the shortest
 program $P$ such that P is a witness to the satisfiability of $\sigma$.
\end{definition}

Let us consider the number of iterations of the CEGIS loop $n$ required for a specification
$\sigma$.  Since we enumerate candidate programs in order of length, we are always synthesising
programs with length no greater than $K(\sigma)$ (since when we enumerate the first correct program,
we will terminate).  So the space of solutions we search over is the space
of functions computed by $\mathcal{L}$-programs of length no greater than $K(\sigma)$.  Let's
denote this set $\mathcal{L}(K(\sigma))$.
Since there are $O(2^{K(\sigma)})$ \emph{programs} of length $K(\sigma)$ and some functions
will be computed by more than one program, we have $| \mathcal{L}(K(\sigma)) | \leq O(2^{K(\sigma)})$.

Each iteration of the CEGIS loop distinguishes at least one incorrect function from the set of correct
functions, so the loop will iterate no more than $| \mathcal{L}(K(\sigma)) |$ times.
Therefore another bound on our runtime is
$NTIME\left(2^{K(\sigma)} \right)$.

\section{Experiments}
\label{sec:kalashnikov-experiments}

We implemented our decision procedure for $SF_\mathcal{D}$ as 
the \kalashnikov tool.
We used \kalashnikov to solve formulae generated from a variety of problems
taken from superoptimisation, code deobfuscation, floating point verification,
ranking function and recurrent set synthesis, safety proving, and bug finding.
The superoptimisation and code deobfuscation benchmarks were taken from the experiments
of~\cite{brahma}; the termination benchmarks were taken from SVCOMP'15~\cite{svcomp15} and
they include the experiments of~\cite{synth-termination}; the safety benchmarks are taken from the
experiments of~\cite{danger-invariants}.

We would like to stress that these experiments serve to evaluate the potential of using
our program synthesiser as a backend for many program analysis tasks and are not intended
to compare performance with specialised solvers for each of these tasks.

We ran our experiments on a 4-core, 3.30\,GHz Core i5 with 8\,GB of RAM.  Each benchmark
was run with a timeout of 180\,s.  The results are shown in Table~\ref{tbl:results}. For each category of benchmarks, we report the total
number of benchmarks in that category, the number we were able to solve within the time limit,
the average solution size (in instructions), the average number of iterations of the CEGIS loop, the average time and total
time taken. The deobfuscation and floating point benchmarks are considered together with the superoptimisation ones.  

It should be understood that
in contrast to less expressive logics that might be invoked several times in the analysis of some
problem, each of these benchmarks is a ``complete'' problem from
the given problem domain.  For example, each of the benchmarks in the termination category
requires \kalashnikov to prove that a full program terminates, i.e.~it must find a ranking function
and supporting invariants, then prove that these constitute a valid termination proof for the program
being analysed.

\paragraph{Discussion of the experimental results.}
The timings show that for the instances where we can find a satisfying assignment,
we tend to do so quite quickly (on the order of a few seconds).  Furthermore the
programs we synthesise are often short, even when the problem domain is very complex, such as for
liveness and safety.

Not all of these benchmarks are satisfiable, and in particular around half of the termination
benchmarks correspond to attempted proofs that non-terminating programs terminate and vice versa.
This illustrates one of the current shortcomings of \kalashnikov as a decision procedure:
we can only conclude that a formula is unsatisfiable once we have examined candidate solutions
up to a very high length bound.  Being able to detect unsatisfiability of a $SF_\mathcal{D}$
formula earlier than this would be extremely valuable.  We note that for some formulae we can
simultaneously search for a proof of satisfiability and of unsatisfiability.  For example,
for some program $P$, we can construct a formula $\phi_T$ that is satisfiable iff $P$ terminates
on all inputs, and another formula $\phi_N$ that is satisfiable iff $P$ does not terminate for some
input~\cite{synth-termination}.  Therefore $\phi_T \vee \phi_N$ is guaranteed to be satisfiable, so we can
synthesise a program proving $\phi_T \vee \phi_N$ and then check which of $\phi_T$ and $\phi_N$ it
proves.  This avoids the bad case where we try to synthesise a solution for an unsatisfiable
specification.

\begin{table}[h]
\resizebox{\textwidth}{!}{
 \begin{tabular}{|l||c|c|c|c|c|c|c|}
 \hline
   Category & \#Benchmarks & \#Solved & Avg. solution size & Avg. iterations & Avg. time (s) & Total time (s) \\
   \hline
   \hline
   Superoptimisation & 29 & 22 & 4.1 & 2.7 & 7.9 & 166.1 \\
   Termination & 47 & 35 & 5.7 & 14.4 & 11.2 & 392.9 \\
   Safety & 20 & 18 & 8.3 & 7.1 & 11.3 & 203.9 \\
   \hline
   \hline
   Total & 96 & 75 & 5.9 & 9.2 & 10.3 & 762.9 \\
   \hline
 \end{tabular}
}
 \caption{Experimental results.\label{tbl:results}}
\end{table}

To help understand the role of the different solvers involved in the synthesis process, we provide
a breakdown of how often each solver ``won'', i.e.~was the first to return an answer.
This breakdown is shown in Table~\ref{tbl:wins}.
We see that GP and explicit account for the great majority of the responses, with the load spread
fairly evenly between them.  This distribution illustrates the different strengths of each solver:
GP is very good at generating candidates, explicit is very good at finding counterexamples
and {\sc CBMC} is very good at proving that candidates are correct.  The GP and explicit numbers are
similar because they are approximately ``number of candidates found'' and ``number of  candidates refuted''
respectively.  The {\sc CBMC} column is approximately ``number of candidates proved correct''.
The spread of winners here shows that each of the search strategies is contributing something to the
overall search and that the strategies are able to co-operate with each other.

\begin{table}
\centering
\subfloat[a][How often each solver ``wins''.]{
\begin{tabular}{|c|c|c|c|}
\hline
 {\sc CBMC} & Explicit & GP & Total \\
 \hline
 1122 & 2499 & 1654 & 5405 \\
 21\% & 46\% & 31\% & 100\% \\
 \hline
\end{tabular}
\label{tbl:wins}
}
%
\qquad\qquad
\subfloat[b][Where the time is spent.]{
\begin{tabular}{|c|c|c|c|}
\hline
{\sc synth} & {\sc verif} & {\sc generalize} & Total \\
\hline
6937\,s & 1114\,s & 559\,s & 8052\,s \\
86\% & 14\% & 7\% & 100\% \\
\hline
\end{tabular}
\label{tbl:time}
}
%
%
\caption{Statistics about the experimental results.}
\end{table}

\begin{table}
\centering
 \begin{tabular}{|l||c|c|c|c|c|c|}
 \hline
   & \#Benchmarks & \#Solved & \#TO & \#Crashes & Avg. time (s) & Spec. size \\
   \hline
  \kalashnikov & 20 & 18 & 2 & 0 & 11.3 & 341 \\
  {\sc eSolver} & 20 & 7 & 5 & 8 & 13.6 & 3140 \\
  \hline
 \end{tabular}

 \caption{Comparison of \kalashnikov and {\sc eSolver} on the safety benchmarks.\label{tbl:sygus}}
\end{table}

To help understand where the time is spent in our solver, Table~\ref{tbl:time} shows how much time is
spent in {\sc synth}, {\sc verif} and constant generalization.  Note that
generalization counts towards {\sc verif}'s time.  We can see that synthesising candidates takes
much longer than verifying them, which suggests that improved procedures for candidate synthesis
will lead to good overall performance improvements.  However, the times shown in this table
include all the runs that timed out, as well as those that succeeded.  We have observed that
runs which time out spend more time in synthesis than runs which succeed, so the distribution here
is biased by the cost of timeouts.

\subsection{Comparison to SyGuS}

In order to compare \kalashnikov to another synthesis based decision procedure, we translated the 20 safety
benchmarks into the SyGuS format~\cite{sygus} and ran the best available SyGuS solver ({\sc eSolver}, taken
from the SyGuS Github repository on 5/7/2015))
on these benchmarks.  We ran {\sc eSolver} on the same machine used for the previous experiments.
The results of these experiments are shown in Table~\ref{tbl:sygus}, which contains the total number
of benchmarks, the number of benchmarks solved correctly, the number of timeouts, the number of crashes
(exceptions thrown by the solver), the mean time to successfully solve and the total number of lines
in the 20 specifications.

Our comparison only uses 20 of the 127 benchmarks because it was difficult for us to convert from our
specification format (a subset of C) into the SyGuS format.  In particular, there are no tools for converting
C programs to SyGuS specifcations and we found it quite difficult to write such a tool so we had to do the conversion by hand.
Since the {\sc eSolver} tool crashed on many of the instances we tried, we reran the experiments on the StarExec platform~\cite{starexec}
to check that we had not made mistakes setting up our environment, however the same instances also
caused exceptions on StarExec.  A technical limitation of SyGuS's output format compared to \kalashnikov's meant that
while \kalashnikov can express lexicographic ranking functions of unbounded dimension, SyGuS cannot.  The benchmarks
we converted to SyGuS involved searching for ranking functions, which meant that we could only use
benchmarks that did not require lexicographic ranking functions.

Overall, we can see that \kalashnikov performs significantly better than the best SyGuS solver on these benchmarks,
which validates our claim that \kalashnikov is a good backend for program analysis problems.  It is also
worth noting that \kalashnikov specifications are significantly more concise than SyGuS specifications, as witnessed
by the total size of the specifications: the \kalashnikov specifications are around 11\% of the size of the
SyGuS ones.

We noticed that for a lot of the cases in which {\sc eSolver} timed out, \kalashnikov found a solution that involved
non-trivial constants.  Since {\sc eSolver} represents constants in unary (as chains of additions),
finding programs containing constants, or finding existentially quantified first order variables is very expensive.
\kalashnikov's strategies for finding and generalising constants make it much more efficient at this subtask.

\section{Conclusions}

We have shown that the synthesis fragment is 
well-suited for program verification by using it to directly encode safety, liveness and superoptimisation
properties. 

We built a decision procedure for $SF_\mathcal{D}$ via a reduction to 
finite state program synthesis. 
The synthesis algorithm is optimised for program analysis and uses a combination of symbolic model checking, explicit
state model checking and stochastic search. 
An important strategy is generalisation -- we find simple solutions that solve a restricted case of the specification, then
 try to generalise to a full solution.
We evaluated the program synthesiser on several static analysis problems, showing the tractability of the approach.  






{\footnotesize
\bibliography{all,thesis}{}
\bibliographystyle{splncs}}

\appendix
\section{Proofs}
\label{app:proofs}

\subsection{$SF_\mathcal{D}$} \label{app:finite-domain}

\begin{theorem}[Fagin's Theorem~\cite{fagin}]
\label{thm:fagin}
 The class of structures $A$ recognisable in time $|A|^k$, for some $k$, by a nondeterministic Turing machine
 is exactly the class of structures definable by existential second-order sentences.
\end{theorem}

\begin{theorem}[$SF_\mathcal{D}$ is NEXPTIME-complete]
 For an instance of Definition~\ref{def:synthesis-frag} with $n$ first-order variables, where the ground terms are interpreted over $\mathcal{D}$,
 checking the truth of the formula is NEXPTIME-complete.
\end{theorem}

\begin{proof}
We will apply Theorem~\ref{thm:fagin}.  To do so we must establish the size of
the universe implied by Theorem~\ref{thm:fagin}.  Since Definition~\ref{def:synthesis-frag} uses
$n$ $\mathcal{D}$ variables, the universe is the set of interpretations
of the $n$ variables.  This set has size $|\mathcal{D}|^n$, and so by Theorem~\ref{thm:fagin},
Definition~\ref{def:synthesis-frag} over finite domains defines exactly the class sets recognisable in $(|\mathcal{D}|^n)^k$
time by a nondeterministic Turing machine.  This is the class NEXPTIME, and
so checking validity of an arbitrary instance of Definition~\ref{def:synthesis-frag} over $\mathcal{D}$
is NEXPTIME-complete.
\end{proof}

\subsection{Program Encodings}
\label{app:encodings}
We encode finite-state programs as loop-free
imperative programs consisting of a sequence of instructions, each instruction
consisting of an opcode and a tuple of operands.  The opcode specifies which
operation is to be performed and the operands are the arguments on which the
operation will be performed.  We allow an operand to be one of: a constant literal,
an input to the program, or the result of some previous instruction.  Such a program
has a natural correspondence with a combinational circuit.

A sequence of instructions is certainly a natural encoding of a program,
but we might wonder if it is the \emph{best} encoding.
We can show that for a reasonable set of instruction types (i.e.~valid opcodes), this encoding
is optimal in a sense we will now discuss.
An encoding scheme $E$ takes a function $f$ and assigns it a name $s$.  For a given ensemble
of functions $F$ we are interested in the worst-case behaviour of the encoding $E$, that is we are
interested in the quantitiy $$|E(F)| = \max \{ |E(f)| \mid f \in F \}$$
If for every encoding $E'$, we have that $$|E(F)| = |E'(F)|$$ then we say that $E$ is an \emph{optimal encoding}
for $F$.  Similarly if for every encoding $E'$, we have $$O(|E(F)|) \subseteq O(|E'(F)|)$$ we say that $E$ is
an \emph{asymptotically optimal encoding} for $F$.

\begin{lemma}[Languages with ITE are Universal and Optimal Encodings for Finite Functions]
\label{lem:ite-size}
 For an imperative programming language including instructions
 for testing equality of two values (EQ) and an if-then-else
 (ITE) instruction, any total function $f : \mathcal{S} \to \mathcal{S}$
 can be computed by a program of size $O(| \mathcal{S} | \log | \mathcal{S} |)$ bits.
\end{lemma}

\begin{proof}
The function $f$ is computed by the following program:

 \begin{verbatim}
t1 = EQ(x, 1)
t2 = ITE(t1, f(1), f(0))
t3 = EQ(x, 2)
t4 = ITE(t3, f(2), t2)
...
 \end{verbatim}

Each operand can be encoded in $\log_2 (| \mathcal{S} | + l) = \log_2 (3 \times | \mathcal{S} |)$ bits.
So each instruction can be encoded in $O(\log | \mathcal{S} |)$ bits and there are $O(|\mathcal{S}|)$
instructions in the program, so the whole program can be encoded in $O(| \mathcal{S} | \log | \mathcal{S} |)$
bits.
\end{proof}

\begin{lemma}
\label{lem:large-encoding}
 Any representation that is capable of encoding an arbitrary total function $f : \mathcal{S} \to \mathcal{S}$
 must require at least $O(| \mathcal{S} | \log | \mathcal{S} |)$ bits to encode some functions.
\end{lemma}

\begin{proof}
 There are $| \mathcal{S} |^{| \mathcal{S} |}$ total functions $f : \mathcal{S} \to \mathcal{S}$.
 Therefore by the pigeonhole principle, any encoding that can encode an arbitrary function must use
 at least $\log_2 (| \mathcal{S} |^{| \mathcal{S} |}) = O(| \mathcal{S} | \log_2 | \mathcal{S} |)$
 bits to encode some function.
\end{proof}

From Lemma~\ref{lem:ite-size} and Lemma~\ref{lem:large-encoding}, we can conclude that \emph{any}
set of instruction types that include ITE is an asymptotically optimal function encoding for total functions with
finite domains.

\subsection{Complexity of Finite State Program Synthesis}

\begin{theorem}[$SF_\mathcal{D}$ is Polynomial Time Reducible to Finite Synthesis]
 Every instance of 
Definition~\ref{def:synthesis-frag}, where the ground terms are interpreted over $\mathcal{D}$
 is polynomial time reducible to an instance
 of Definition~\ref{def:finite-synth-formula}.
\end{theorem}


\begin{proof}
 We first Skolemise the instance of definition~\ref{def:synthesis-frag} to produce an equisatisfiable
 second-order sentence with the first-order part only having universal quantifiers
 (i.e. bring the formula into Skolem normal form).  This process will have introduced
 a function symbol for each first order existentially quantified variable and taken
 linear time.  Now we just existentially quantify over the Skolem functions, which
 again takes linear time and space.
 The resulting formula is an instance of Definition~\ref{def:finite-synth-formula}.
\end{proof}

\subsection{Soundness and Completeness}

\begin{theorem}\label{thm:synth-sound}
Algorithm~\ref{alg:cegis} is sound -- if it terminates with witness $P$, then
$P \models \sigma$.
\end{theorem}

\begin{proof}
 The procedure {\sc synth} terminates only if {\sc synth} returns ``valid''.  In that
 case, $\exists x . \lnot \sigma(P, x)$ is unsatisfiable and so $\forall x . \sigma(P, x)$ holds.
\end{proof}

\begin{lemma}
 \label{thm:synth-semi-complete}
 Algorithm~\ref{alg:cegis} is semi-complete -- if a solution $P \models \sigma$
 exists then Algorithm~\ref{alg:cegis} will find it.
\end{lemma}

\begin{proof}
 If the domain $X$ is finite then the loop in procedure {\sc synth} can only
 iterate $| X |$ times, since by this time all of the elements of $X$ would have been
 added to the inputs set.  Therefore if the {\sc synth} procedure always terminates,
 Algorithm~\ref{alg:cegis} does as well.

 Since the {\sc ExplicitSearch} routine enumerates all programs (as can be seen by induction on
 the program length $l$), it will eventually enumerate a program that meets the specification
 on whatever set of inputs are currently being tracked, since by assumption such a program
 exists.  Since the first-order theory is
 decidable, the query in {\sc verif} will succeed for this program, causing the algorithm to terminate.
 The set of correct programs is therefore recursively enumerable and Algorithm~\ref{alg:cegis}
 enumerates this set, so it is semi-complete.
\end{proof}

\begin{theorem}
 \label{thm:finite-synth-complete}
 Algorithm~\ref{alg:cegis} with the stopping condition described in Section~\ref{sec:stopping-condition}
 is complete when instantiated with \newC as a background theory -- it will terminate for all specifications $\sigma$.
\end{theorem}

\begin{proof}
 If the specification is satisfiable then Theorem~\ref{thm:synth-semi-complete} holds, and if it is not
 then the stopping condition will eventually hold at which point we (correctly) terminate with an ``unsatisfiable'' verdict.
\end{proof}

\section{Counterexample Guided Inductive Synthesis (CEGIS)} 
\label{app:cegis}

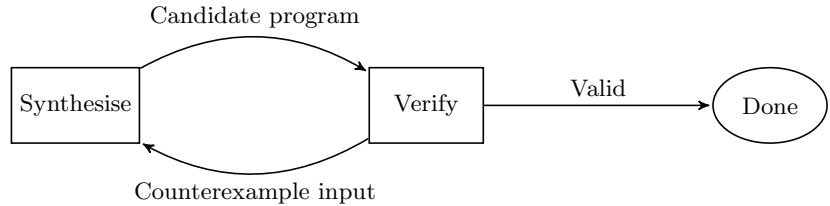
\begin{figure}
 \centering
 \begin{tikzpicture}[scale=0.5,->,>=stealth',shorten >=1pt,auto,
 semithick, initial text=]

  \matrix[nodes={draw, fill=none, scale=1, shape=rectangle, minimum height=1cm, minimum width=1.5cm},
          row sep=2cm, column sep=3cm] {
   \node (synth) {Synthesise};
   &
   \node (verif) {Verify}; 
   &
   \node[ellipse] (done) {Done}; \\
  };

   \path
    (synth) edge [bend left] node {Candidate program} (verif)
    (verif) edge [bend left] node {Counterexample input} (synth)
    (verif) edge node {Valid} (done);
 \end{tikzpicture}
 
 \caption{Abstract synthesis refinement loop
 \label{fig:abstract-refinement}}
\end{figure}

\section{Encoding in \newC}

\begin{figure}
\begin{center}
{\small

\setlength{\tabcolsep}{14pt}
Integer arithmetic instructions:

\begin{tabular}{llll}
 \verb|add a b| & \verb|sub a b| & \verb|mul a b| & \verb|div a b| \\
 \verb|neg a| &   \verb|mod a b| & \verb|min a b| & \verb|max a b|
\end{tabular}

\medskip

Bitwise logical and shift instructions:

\begin{tabular}{lll}
 \verb|and  a b| & \verb|or   a b| & \verb|xor a b| \\
 \verb|lshr a b| & \verb|ashr a b| & \verb|not a|
\end{tabular}

\medskip

Unsigned and signed comparison instructions:

\begin{tabular}{lll}
 \verb|le  a b| & \verb|lt  a b| & \verb|sle  a b| \\
 \verb|slt a b| & \verb|eq  a b| & \verb|neq  a b| \\
\end{tabular}

\medskip

Miscellaneous logical instructions:

\begin{tabular}{lll}
 \verb|implies a b| & \verb|ite a b c| &  \\
\end{tabular}

\medskip

\setlength{\tabcolsep}{12pt}

Floating-point arithmetic:

\begin{tabular}{llll}
 \verb|fadd a b| & \verb|fsub a b| & \verb|fmul a b| & \verb|fdiv a b|
\end{tabular}

}
\end{center}

 \caption{The language $\mathcal{L}$}
 \label{fig:l-language}
\end{figure}




  



\end{document}